\documentclass[12pt]{article}

\usepackage[affil-it]{authblk}

\usepackage{enumerate}

\usepackage{tikz}
\usetikzlibrary{chains}
\usetikzlibrary{decorations.pathreplacing}

\usepackage[margin=1in]{geometry}

\usepackage{amsmath,amsfonts,amssymb,times,latexsym}
\usepackage{amsthm}
\theoremstyle{plain}
\newtheorem{theorem}{Theorem}
\newtheorem{lemma}[theorem]{Lemma}
\newtheorem{proposition}[theorem]{Proposition}

\theoremstyle{definition}
\newtheorem{definition}[theorem]{Definition}

\newcommand{\abs}[1]{\left\lvert#1\right\rvert}
\newcommand{\sturmian}[1]{\textbf{c}_{#1}}
\newcommand{\morphism}[1]{\varphi_{#1}}
\newcommand{\ostrowski}[2]{\textsc{OR}_{#1}(#2)}

\begin{document}

\title{Ostrowski Numeration and the Local Period of Sturmian Words}

\author{Luke Schaeffer}
\affil{School of Computer Science \\
University of Waterloo \\
Waterloo, ON  N2L 3G1 Canada \\
\texttt{l3schaef@cs.uwaterloo.ca}}


\maketitle

\begin{abstract}
We show that the local period at position $n$ in a characteristic Sturmian word can be given in terms of the Ostrowski representation for $n+1$. 
\end{abstract}

\section{Introduction}

We consider characteristic Sturmian words, which are infinite words over $\{ 0, 1 \}$ such that the $i$th character is 
$$
\lfloor \alpha(i+1) \rfloor - \lfloor \alpha i \rfloor - \lfloor \alpha \rfloor
$$
for some irrational $\alpha$. We give an alternate definition later better suited to our purposes. Let $f_{w}(n)$ denote the number of factors of length $n$ in $w$, also known as the \emph{subword complexity} of $O(n)$. It is well-known that $f_{w}(n) = n+1$ when $w$ is a Sturmian word. On the other hand, the Coven-Hedlund theorem \cite{covenhedlund} states that $f_{w}(n)$ is either bounded or $f_{w}(n) \geq n+1$ for all $n$. In this sense, Sturmian words are extremal with respect to subword complexity. 

In a recent paper \cite{restivomignosi}, Restivo and Mignosi show that characteristic Sturmian words are also extremal with respect to local period, which we define shortly as part of Definition~\ref{def:localperiod}. Let $p_{w}(n)$ denote the local period of a word $w$ at position $n$. The critical factorization theorem states that either $p_{w}(n)$ is bounded or $p_{w}(n) \geq n+1$ for infinitely many $n$. Restivo and Mignosi show that when $w$ is a characteristic Sturmian word, $p_{w}(n)$ is at most $n+1$ and $p_{w}(n) = n+1$ infinitely often. Hence, characteristic Sturmian words also have extremal local periods. 

Unlike subword complexity, the local period function $p_{w}(n)$ is erratic. Consider Table~\ref{tab:fibperiod}, which gives the local period at points in $F$, the Fibonacci word. 
\begin{table}[h]
\begin{center}
\begin{tabular}{|c|ccccccccccccccccccccc|}
\hline
$n$        & 0 & 1 & 2 & 3 & 4 & 5 & 6 & 7 & 8 & 9 & 10 & 11 & 12 & 13 & 14 & 15 & 16 & 17 & 18 & 
19 & 20 \\
\hline
$p_{F}(n)$ & 1 & 2 & 3 & 1 & 5 & 2 & 2 & 8 & 1 & 3 &  3 &  1 & 13 &  2 &  2 &  5 &  1 &  5 &  2 & 
 2 & 21 \\
\hline
\end{tabular}
\end{center}
\caption{The local period function for the Fibonacci word.}
\label{tab:fibperiod}
\end{table}
Although there are patterns in the table (for example, each $p_{F}(n)$ is a Fibonacci number), it is not obvious how $p_{F}(n)$ is related to $n$ in general. Shallit \cite{shallit} showed that $p_{F}(n)$ is easily computed from the Zeckendorf representation of $n+1$, and conjectured that for a general characteristic Sturmian word $w$, $p_{w}(n)$ is a simple function of the corresponding Ostrowski representation for $n+1$. In this paper, we confirm Shallit's conjecture by describing $p_{w}(n)$ in terms of the Ostrowski representation for $n+1$. 

\section{Notation}

Let $\Sigma := \{ 0, 1 \}$ for the rest of this paper. We write $w[n]$ to denote the $n$th letter of a word $w$ (finite or infinite), and $w[i..j]$ for the factor $w[i] w[i+1] \cdots w[j-1] w[j]$. We use the convention that the first character in $w$ is $w[0]$. Let $\abs{w}$ denote the length of a finite word $w$. 

\subsection{Repetition words}

\begin{definition}
Let $w$ be an infinite word over a finite alphabet $\Sigma$. A \emph{repetition word in $w$ at position $i$} is a non-empty factor $w[i..j]$ such that either $w[i..j]$ is a prefix of $w[0..i-1]$ or $w[0..i-1]$ is a prefix of $w[i..j]$.
\end{definition}
If the infinite word $w$ is recurrent (i.e., every factor in $w$ occurs more than once in $w$) then every factor occurs infinitely many times. In particular, for every $i$ the prefix $w[0..i-1]$ occurs in $w[i..\infty]$, so there exists a repetition word at every position in a recurrent word. 
\begin{definition}
\label{def:localperiod}
Let $w$ be an infinite recurrent word over a finite alphabet $\Sigma$. Let $r_{w}(i)$ denote the shortest repetition word in $w$ at position $i$. The length of the shortest repetition word, denoted by $p_{w}(i) := \abs{r_{w}(i)}$, is called the \emph{local period in $w$ at position $i$}. 
\end{definition}
We note that Sturmian words are recurrent, so $p_{w}(i)$ and $r_{w}(i)$ exist at every position for a characteristic Sturmian word $w$. We omit further discussion of the existence of $p_{w}(i)$ and $r_{w}(i)$. 

For example, consider the Fibonacci word $F$ shown in Figure~\ref{fig:repetition}. The factors $F[5..6] = 01$, $F[5..9] = 01001$ and $F[5..17] = 0100100101001$ are examples of repetition words in the Fibonacci word at position 5. The shortest repetition word at position $5$ is $r_{F}(5) = F[5..6] = 01$ and therefore the local period at position 5 is $p_{F}(5) = 2$. 

\begin{figure}[h]
\begin{center}
\begin{tikzpicture}
\small
\edef\sizetape{5.5mm}
\tikzstyle{tape}=[draw,minimum size=\sizetape]
\tikzstyle{nums}=[minimum size=\sizetape,font=\scriptsize]

\begin{scope}[start chain=1 going right, node distance=-0.15mm]
		\node [on chain=1] {$F = $};
    \node [on chain=1,tape] (start) {0};
    \node [on chain=1,tape] {1};
    \node [on chain=1,tape] {0};
    \node [on chain=1,tape] {0};
    \node [on chain=1,tape] (prev) {1};
    \node [on chain=1,tape,xshift=0.5mm] (next) {0};
    \node [on chain=1,tape] {1};
    \node [on chain=1,tape] {0};
    \node [on chain=1,tape] {0};
    \node [on chain=1,tape] {1};
    \node [on chain=1,tape] {0};
    \node [on chain=1,tape] {0};
    \node [on chain=1,tape] {1};
    \node [on chain=1,tape] {0};
    \node [on chain=1,tape] {1};
    \node [on chain=1,tape] {0};
    \node [on chain=1,tape] {0};
    \node [on chain=1,tape] {1};
    \node [on chain=1,tape] {0};
    \node [on chain=1,tape] {1};
    \node [on chain=1,tape] {0};
    \node [on chain=1,tape,draw=none] {$\ldots$};
\end{scope}

\begin{scope}[start chain=1 going right, node distance=-0.15mm]
    \node [on chain=1,nums,yshift=6mm] at (start) {0};
    \node [on chain=1,nums] {1};
    \node [on chain=1,nums] {2};
    \node [on chain=1,nums] {3};
    \node [on chain=1,nums] {4};
    \node [on chain=1,nums,xshift=0.5mm] {5};
    \node [on chain=1,nums] {6};
    \node [on chain=1,nums] {7};
    \node [on chain=1,nums] {8};
    \node [on chain=1,nums] {9};
    \node [on chain=1,nums] {10};
    \node [on chain=1,nums] {11};
    \node [on chain=1,nums] {12};
    \node [on chain=1,nums] {13};
    \node [on chain=1,nums] {14};
    \node [on chain=1,nums] {15};
    \node [on chain=1,nums] {16};
    \node [on chain=1,nums] {17};
    \node [on chain=1,nums] {18};
    \node [on chain=1,nums] {19};
    \node [on chain=1,nums] {20};
\end{scope}

\begin{scope}[start chain=1 going left, node distance=-0.15mm]
    \node [on chain=1,tape,yshift=-6mm] at (prev) {1};
    \node [on chain=1,tape] {0};
\end{scope}
\begin{scope}[start chain=1 going right, node distance=-0.15mm]
    \node [on chain=1,tape,yshift=-6mm] at (next) {0};
    \node [on chain=1,tape] {1};
\end{scope}

\begin{scope}[start chain=1 going left, node distance=-0.15mm]
    \node [on chain=1,tape,yshift=-12mm] at (prev) {1};
    \node [on chain=1,tape] {0};
    \node [on chain=1,tape] {0};
    \node [on chain=1,tape] {1};
    \node [on chain=1,tape] {0};
\end{scope}
\begin{scope}[start chain=1 going right, node distance=-0.15mm]
    \node [on chain=1,tape,yshift=-12mm] at (next) {0};
    \node [on chain=1,tape] {1};
    \node [on chain=1,tape] {0};
    \node [on chain=1,tape] {0};
    \node [on chain=1,tape] {1};
\end{scope}

\begin{scope}[start chain=1 going left, node distance=-0.15mm]
    \node [on chain=1,tape,yshift=-18mm] at (prev) {1};
    \node [on chain=1,tape] {0};
    \node [on chain=1,tape] {0};
    \node [on chain=1,tape] {1};
    \node [on chain=1,tape] {0};
    \node [on chain=1,tape] {1};
    \node [on chain=1,tape] {0};
    \node [on chain=1,tape] {0};
    \node [on chain=1,tape] {1};
    \node [on chain=1,tape] {0};
    \node [on chain=1,tape] {0};
    \node [on chain=1,tape] {1};
    \node [on chain=1,tape] {0};
\end{scope}
\begin{scope}[start chain=1 going right, node distance=-0.15mm]
    \node [on chain=1,tape,yshift=-18mm] at (next) {0};
    \node [on chain=1,tape] {1};
    \node [on chain=1,tape] {0};
    \node [on chain=1,tape] {0};
    \node [on chain=1,tape] {1};
    \node [on chain=1,tape] {0};
    \node [on chain=1,tape] {0};
    \node [on chain=1,tape] {1};
    \node [on chain=1,tape] {0};
    \node [on chain=1,tape] {1};
    \node [on chain=1,tape] {0};
    \node [on chain=1,tape] {0};
    \node [on chain=1,tape] {1};
\end{scope}
\end{tikzpicture}
\end{center}
\caption{The Fibonacci word $F$ and some repetition words at position $5$}
\label{fig:repetition}
\end{figure}

\section{Characteristic Sturmian Words and the Ostrowski Representation}

We define characteristic Sturmian words and the Ostrowski representation based on directive sequences of integers, defined below. For every directive sequence there is a corresponding characteristic Sturmian word. Similarly, for each directive sequence there is an Ostrowski representation associating nonnegative integers with strings. 
\begin{definition}
A \emph{directive sequence} $\alpha = \{ a_i \}_{i=0}^{\infty}$ is a sequence of nonnegative integers, where $a_i > 0$ for all $i > 0$. 
\end{definition}
Directive sequences are in some sense infinite words over the natural numbers, so we use the same indexing/factor notation. The notation $\alpha[i]$ indicates the $i$th term, $a_i$. We will frequently separate a directive sequence $\alpha$ into the first term, $\alpha[0]$, and the rest of the sequence, $\alpha[1..\infty]$.

Note that our definitions for characteristic Sturmian words and Ostrowski representations deviate slightly from the definitions given in our references, \cite{shallitallouche} and \cite{berstel}. Specifically, there are two main differences between our definition and \cite{shallitallouche}:
\begin{enumerate}
\item We start indexing the directive sequence at zero instead of one. 
\item The first term is interpreted differently. For example, if the first term in the sequence $a$ then our characteristic Sturmian word begins with $0^a 1$, whereas the characteristic Sturmian word in \cite{shallitallouche} begins with $0^{a-1} 1$. 
\end{enumerate}
In other words, we are describing the same mathematical objects, but label them with slightly different directive sequences. Any result that does not explicitly reference the terms of the directive sequence will be true for either set of definitions. This includes our main result, Theorem~\ref{theorem:main}.

\subsection{Characteristic Sturmian Words}

Consider the following collection of morphisms.
\begin{definition}
For each $k \geq 0$, we define a morphism $\morphism{k} \colon \Sigma^{*} \rightarrow \Sigma^{*}$ such that
\begin{align*}
\morphism{k}(0) &= 0^{k} 1 \\
\morphism{k}(1) &= 0
\end{align*} 
for all $k \geq 0$.
\end{definition}
Given a directive sequence, we use this collection of morphisms to construct a sequence of words. 
\begin{definition}
Let $\alpha$ be a directive sequence. We define a sequence of finite words $\{ X_i \}_{i=0}^{\infty}$ over $\Sigma$ where 
\begin{align*}
X_n &= (\morphism{\alpha[0]} \circ \cdots \circ \morphism{\alpha[n-1]})(0).
\end{align*}
We call $\{ X_i \}_{i=0}^{\infty}$ the \emph{standard sequence}, and we say $X_i$ is the \emph{$i$th characteristic block}. 
\end{definition}
Sometimes the characteristic blocks are defined recursively as follows. 
\begin{proposition}
\label{prop:recurse}
Let $\alpha$ be a directive sequence and let $\{ X_i \}_{i=0}^{\infty}$ be the corresponding directive sequence. Then
\begin{align*}
X_{n} &= \begin{cases}
0, & \text{if $n = 0$;} \\
0^{\alpha[0]} 1, & \text{if $n = 1$;} \\
X_{n-1}^{\alpha[n-1]} X_{n-2}, & \text{if $n \geq 2$.}
\end{cases}
\end{align*}
\end{proposition}
\begin{proof}
See Theorem 9.1.8 in \cite{shallitallouche}. Note that due to a difference in definitions, the authors number the directive sequence starting from one instead of zero, and they treat the first term differently (i.e., they define $X_1$ as $0^{a_1 - 1} 1$ instead of $0^{\alpha[0]} 1$).
\end{proof}

It follows from the proposition that $X_{n-1}$ is a prefix of $X_{n}$ for each $n \geq 2$, and therefore the limit $\lim_{n \rightarrow \infty} X_n$ exists. We define $\sturmian{\alpha}$, the characteristic Sturmian word corresponding to the directive sequence $\alpha$, to be this limit. 
$$
\sturmian{\alpha} := \lim_{n \rightarrow \infty} X_n.
$$
Then $X_n$ is a prefix of $\sturmian{\alpha}$ for each $n \geq 2$. 

There is a simple relationship between $\sturmian{\alpha}$, $\alpha[0]$ and $\sturmian{\alpha[1..\infty]}$, given in the following proposition.
\begin{proposition}
\label{prop:image}
Let $\alpha$ be a directive sequence, and let $\beta := \alpha[1..\infty]$. Then 
\begin{align*}
\sturmian{\alpha} &= \morphism{\alpha[0]} \left( \sturmian{\beta} \right)
\end{align*}
\end{proposition} 
\begin{proof}{(Sketch)}
We factor $\morphism{\alpha[0]}$ out of each $X_i$ and then out of the limit. 
$$
\sturmian{\alpha} = \lim_{n \rightarrow \infty} (\morphism{\alpha[0]} \circ \cdots \circ \morphism{\alpha[n-1]})(0) = \morphism{\alpha[0]} \left( \lim_{n \rightarrow \infty} (\morphism{\alpha[1]} \circ \cdots \circ \morphism{\alpha[n-1]})(0) \right) = \morphism{\alpha[0]} \left( \sturmian{\beta} \right).
$$
Alternatively, see Theorem 9.1.8 in \cite{shallitallouche} for a similar result. 
\end{proof}
Notice that if $\alpha[0] = 0$ then $\sturmian{\alpha}$ and $\sturmian{\beta}$ are the same infinite word up to permutation of the alphabet, since $\morphism{0}$ swaps 0 and 1. Permuting the alphabet does not affect the local period or repetition words, so henceforth we assume that the first term of any directive sequence is positive (and therefore all terms are positive). Consequently, all characteristic Sturmian words we consider will start with 0 and avoid the factor $11$. 

Let us give an example of a characteristic Sturmian word. Consider the directive sequence $\alpha$ beginning $1,3,2,2$. Then we can compute the first five terms of the standard sequence 
\begin{align*}
X_0 &= 0 \\
X_1 &= 01 \\
X_2 &= 0101010 \\
X_3 &= 0101010010101001 \\
X_4 &= 010101001010100101010100101010010101010.
\end{align*}
We know $X_4$ is a prefix of $\sturmian{\alpha}$, so we can deduce the first $\abs{X_4} = 39$ characters of $\sturmian{\alpha}$. Thus,
\begin{align*}
\sturmian{\alpha} &= 010101001010100101010100101010010101010 \cdots
\end{align*}
By Proposition~\ref{prop:image}, $\sturmian{\alpha}$ is equal to $\morphism{1}(\sturmian{\alpha[1..\infty]})$. 
\begin{align*}
\sturmian{\alpha} &= 01 \, 01 \, 01 \, 0 \, 01 \, 01 \, 01 \, 0 \, 01 \, 01 \, 01 \, 01 \, 0 \, 01 \, 01 \, 01 \, 0 \, 01 \, 01 \, 01 \, 01 \, 0 \cdots \\
&= \morphism{1}( 0 0 0 1 0 0 0 1 0 0 0 0 1 0 0 0 1 0 0 0 0 1 \cdots ).
\end{align*}

\subsection{Ostrowski representation}

For each directive sequence $\alpha$, there is a corresponding characteristic Sturmian word $\sturmian{\alpha}$. For each characteristic Sturmian word there is a numeration system, the Ostrowski representation, which is closely related to the standard sequence. For example, if the directive sequence is $\alpha = 1, 1, 1, \ldots$ then $\sturmian{\alpha}$ is $F$, the Fibonacci word. The Ostrowski representation for $\alpha = 1,1,1,\ldots$ is the Zeckendorf representation, where we write an integer as a sum of Fibonacci numbers. See chapter three in \cite{shallitallouche} for a description of these numeration systems, but note that their definition of Ostrowski representation differs from our definition.

\begin{definition}
Let $\alpha$ be a directive sequence, and let $\{ X_i \}_{i=0}^{\infty}$ be the corresponding standard sequence. Define an integer sequence $\{ q_{i} \}_{i=0}^{\infty}$ where $q_{i} = \abs{X_i}$ for all $i \geq 0$. Let $n \geq 0$ be an integer. An \emph{$\alpha$-Ostrowski representation} (or simply \emph{Ostrowski representation} when $\alpha$ is understood) for $n$ is a sequence of non-negative integers $\{ d_i \}_{i=0}^{\infty}$ such that 
\begin{enumerate}
\item Only finitely many $d_i$ are nonzero. 
\item $n = \sum_{i} d_i q_i$
\item $0 \leq d_i \leq \alpha[i]$ for all $i \geq 0$. 
\item If $d_{i} = \alpha[i]$ then $d_{i-1} = 0$ for all $i \geq 1$. 
\end{enumerate}
\end{definition}
Note that by Proposition~\ref{prop:recurse}, we can also generate $\{ q_{i} \}_{i=0}^{\infty}$ directly from $\alpha$ using the following recurrence 
\begin{align*}
q_{n} &= \begin{cases}
1, & \text{if $n = 0$;} \\
\alpha[0] + 1, & \text{if $n = 1$;} \\
q_{n-1} \alpha[n-1] + q_{n-2}, & \text{if $n \geq 2$.}
\end{cases}
\end{align*}

It is well-known that for any given directive sequence, there is a unique Ostrowski representation, which we denote $\ostrowski{\alpha}{n}$, for every non-negative integer  \cite{shallitallouche}. Also note that formally $\ostrowski{\alpha}{n}$ is an infinite sequence $\{ d_i \}_{i=0}^{\infty}$, but we often write the terms up to the last nonzero term, e.g., $d_{k} d_{k-1} \cdots d_{1} d_{0}$, with the understanding that $d_{i} = 0$ for $i > k$. This is analogous to decimal representation of integers, where we write the least significant digit last and omit leading zeros.

\begin{theorem}
\label{thm:decomp}
Let $\alpha$ be a directive sequence. Let $n \geq 0$ be an integer, and let $d_{k} d_{k-1} \cdots d_{1} d_{0}$ be an Ostrowski representation for $n$. Then
$$
w := X_{k}^{d_{k}} X_{k-1}^{d_{k-1}} \cdots X_{1}^{d_1} X_{0}^{d_0}
$$
is a proper prefix of $X_{k+1}$, and therefore $w$ is a prefix of $\sturmian{\alpha}$. Since $\abs{w} = \sum_{i} d_{k} \abs{X_i} = n$, it follows that $w = \sturmian{\alpha}[0..n-1]$. 
\end{theorem}
\begin{proof}
This is essentially Theorem 9.1.13 in \cite{shallitallouche}.
\end{proof}

The following technical lemma relates Ostrowski representations for $\alpha$ and $\alpha[1..\infty]$, in much the same way that Proposition~\ref{prop:image} relates $\sturmian{\alpha}$ to $\sturmian{\alpha[1..\infty]}$. 
\begin{lemma}
\label{lemma:image}
Let $\alpha$ be a directive sequence and define $\beta := \alpha[1..\infty]$. Let $n \geq 0$ be an integer with Ostrowski representation $\ostrowski{\alpha}{n} = d_k \cdots d_{0}$. Then there exists an integer $m \geq 0$ such that $\ostrowski{\beta}{m} = d_{k} \cdots d_{1}$ and
$$
\sturmian{\alpha}[0..n-1] = \morphism{\alpha[0]}(\sturmian{\beta}[0..m-1]) 0^{d_0}.
$$
Furthermore, if $d_0 > 0$ then $\sturmian{\beta}[m] = 0$. 
\end{lemma}
\begin{proof}
We leave it to the reader to show that if $d_{k} \cdots d_{0}$ is an $\alpha$-Ostrowski representation then $d_{k} \cdots d_{1}$ is a $\beta$-Ostrowski representation, and conversely, if $d_{k} \cdots d_{1}$ is a $\beta$-Ostrowski representation then $d_{k} \cdots d_{1} 0$ is an $\alpha$-Ostrowski representation. Theorem~\ref{thm:decomp} proves that 
$$\sturmian{\alpha}[0..n-1] = X_{k}^{d_{k}} X_{k-1}^{d_{k-1}} \cdots X_{1}^{d_1} X_{0}^{d_0} = \sturmian{\beta}[0..m-1] 0^{d_0}.$$

Finally, suppose that $d_{0} > 0$ and $\sturmian{\beta}[m] = 1$ for a contradiction. We consider the integer $n - d_{0} + 1$ and its Ostrowski representations. On the one hand, $d_{k} \cdots d_{1} 1$ is a valid Ostrowski representation and $d_{0} - 1$ less than $n$. On the other hand, $$\sturmian{\alpha}[0..n - d_0] = \morphism{\alpha[0]}(\sturmian{\beta}[0..m-1]) 0 = \morphism{\alpha[0]}(\sturmian{\beta}[0..m]),$$
so $\ostrowski{\beta}{m+1}$ followed by $0$ is another Ostrowski representation for $n - d_0 + 1$. This contradicts the uniqueness of Ostrowski representations. 
\end{proof}

Let us continue our earlier example, where we had a directive sequence $\alpha$ beginning $1,3,2,2$. We can compute the first five terms of $\{ q_i \}_{i=0}^{\infty}$. 
\begin{align*}
q_{0} &= \abs{X_0} = 1 \\
q_{1} &= \abs{X_1} = 2 \\
q_{2} &= \abs{X_2} = 7 \\
q_{3} &= \abs{X_3} = 16 \\
q_{4} &= \abs{X_4} = 39.
\end{align*}
In Table~\ref{tab:ostrowski}, we show Ostrowski representations for some small integers.
\begin{table}[h]
\begin{center}
\begin{tabular}{|cc|cc|cc|cc|}
\hline
$n$  & $\ostrowski{\alpha}{n}$ & $n$ & $\ostrowski{\alpha}{n}$ & $n$ & $\ostrowski{\alpha}{n}$ & $n$ & $\ostrowski{\alpha}{n}$ \\
\hline
$0$  &   $0$ & $15$ &  $201$ & $30$ &  $1200$ & $45$ & $10030$ \\
$1$  &   $1$ & $16$ & $1000$ & $31$ &  $1201$ & $46$ & $10100$ \\
$2$  &  $10$ & $17$ & $1001$ & $32$ &  $2000$ & $47$ & $10101$ \\
$3$  &  $11$ & $18$ & $1010$ & $33$ &  $2001$ & $48$ & $10110$ \\
$4$  &  $20$ & $19$ & $1011$ & $34$ &  $2010$ & $49$ & $10111$ \\
$5$  &  $21$ & $20$ & $1020$ & $35$ &  $2011$ & $50$ & $10120$ \\
$6$  &  $30$ & $21$ & $1021$ & $36$ &  $2020$ & $51$ & $10121$ \\
$7$  & $100$ & $22$ & $1030$ & $37$ &  $2021$ & $52$ & $10130$ \\
$8$  & $101$ & $23$ & $1100$ & $38$ &  $2030$ & $53$ & $10200$ \\
$9$  & $110$ & $24$ & $1101$ & $39$ & $10000$ & $54$ & $10201$ \\
$10$ & $111$ & $25$ & $1110$ & $40$ & $10001$ & $55$ & $11000$ \\
$11$ & $120$ & $26$ & $1111$ & $41$ & $10010$ & $56$ & $11001$ \\
$12$ & $121$ & $27$ & $1120$ & $42$ & $10011$ & $57$ & $11010$ \\
$13$ & $130$ & $28$ & $1121$ & $43$ & $10020$ & $58$ & $11011$ \\
$14$ & $200$ & $29$ & $1130$ & $44$ & $10021$ & $59$ & $11020$ \\
\hline
\end{tabular}
\end{center}
\caption{Ostrowski representations where $\alpha = 1,3,2,2,\cdots$}
\label{tab:ostrowski}
\end{table}
By Theorem~\ref{thm:decomp}, we should be able to decompose $\sturmian{\alpha}[0..20]$ as $X_3 X_1^2 X_0$ since $\ostrowski{\alpha}{21} = 1021$.
\begin{align*}
\sturmian{\alpha}[0..20] &= 010101001010100101010 \\
&= (0101010010101001) (01)^2 0 \\
&= X_3 X_1^{2} X_0.
\end{align*}

\section{Local periods in characteristic Sturmian words}

Let $\alpha$ be a directive sequence. Let $p_{\alpha}(n) := p_{\sturmian{\alpha}}(n)$ and $r_{\alpha}(n) := r_{\sturmian{\alpha}}(n)$ be notation for the local period and shortest repetition word for characteristic Sturmian words. In this section we discuss how $p_{\alpha}(n)$ and $r_{\alpha}(n)$ are related to $\ostrowski{\alpha}{n+1}$. 

\begin{definition}
Let $x, y$ be words in $\Sigma^{*}$. Then $x$ is a \emph{conjugate} of $y$ if there exist words $u, v \in \Sigma^{*}$ such that $x = uv$ and $y = vu$.
\end{definition}

\begin{lemma}
\label{lemma:recursivecase}
Let $\alpha$ be a directive sequence, let $\beta := \alpha[1..\infty]$ and $k := \alpha[0]$. Suppose we have integers $m, n \geq 0$ such that $\sturmian{\alpha}[0..n] = \morphism{k}(\sturmian{\beta}[0..m])$. Then 
\begin{enumerate}[(i)]
\item If $u$ is a repetition word in $\sturmian{\beta}$ at position $m$ then there exists a repetition word $v$ in $\sturmian{\alpha}$ at position $n$ such that $\morphism{k}(u)$ is a conjugate of $v$. 
\item If $v$ is a repetition word in $\sturmian{\alpha}$ at position $n$ then there exists a repetition word $u$ in $\sturmian{\beta}$ at position $m$ such that $\morphism{k}(u)$ is a conjugate of $v$.
\end{enumerate}
In particular, $r_{\alpha}(n)$ is a conjugate of $\morphism{k}(r_{\beta}(m))$ when $\sturmian{\alpha}[0..n] = \morphism{k}(\sturmian{\beta}[0..m])$. 
\end{lemma}

\begin{proof}
We divide into two cases based on whether $\sturmian{\beta}[m]$ is $0$ or $1$. The situation when $\sturmian{\beta}[m] = 0$ is shown in Figure~\ref{fig:case0}, and $\sturmian{\beta}[m] = 1$ is shown in Figure~\ref{fig:case1}. These figures, along with the more detailed diagrams in Figures~\ref{fig:detailed0} and \ref{fig:detailed1} later in the proof, indicate how $\morphism{k}$ maps blocks in $\sturmian{\beta}$ to blocks in $\sturmian{\alpha}$. 

\begin{figure}[h]
\begin{minipage}{0.5\textwidth}
\centering
\begin{tikzpicture}
\edef\sizetape{6.5mm}
\tikzstyle{tape}=[draw,minimum size=\sizetape]

\begin{scope}[start chain=1 going right, node distance=-0.15mm]
	  \node [on chain=1] {$\sturmian{\beta} = $} ;
    \node [on chain=1,tape,minimum width=2cm] (start) {};
    \node [on chain=1,tape] (topzero) {$0$};
    \node [on chain=1,tape,draw=none] {$\ldots$};
\end{scope}

\begin{scope}[start chain=1 going right, node distance=-0.15mm]
    \node [on chain=1,yshift=-2cm,xshift=-.5cm] {$\sturmian{\alpha} = $};
    \node [on chain=1,tape,minimum width=3cm] (imagestart) {};
    \node [on chain=1,tape,minimum width=1cm] (imagezero) {$0^{k} 1$};
    \node [on chain=1,tape,draw=none] {$\ldots$};
\end{scope}

\path[dashed] (start.south west) edge (imagestart.north west)
		  (topzero.south west) edge (imagezero.north west)
		  (topzero.south east) edge (imagezero.north east);
	
\draw [thick, decoration={brace,raise=1mm}, decorate]
      (start.north west) -- node[above,yshift=1mm] {$\sturmian{\beta}[0..m]$} (topzero.north east);
		  
\draw [thick, decoration={brace,mirror,raise=1mm}, decorate] 
			(imagestart.south west) -- node[below,yshift=-1mm] {$\sturmian{\alpha}[0..n]$} (imagezero.south east);
\end{tikzpicture}
\caption{Simple diagram for $\sturmian{\beta}[m] = 0$}
\label{fig:case0}
\end{minipage}
\begin{minipage}{0.5\textwidth}
\centering
\begin{tikzpicture}
\edef\sizetape{6.5mm}
\tikzstyle{tape}=[draw,minimum size=\sizetape]

\begin{scope}[start chain=1 going right, node distance=-0.15mm]
	  \node [on chain=1] {$\sturmian{\beta} = $} ;
    \node [on chain=1,tape,minimum width=2cm] (start) {};
    \node [on chain=1,tape] (topzero) {$0$};
    \node [on chain=1,tape] (topone) {$1$};
    \node [on chain=1,tape,draw=none] {$\ldots$};
\end{scope}

\begin{scope}[start chain=1 going right, node distance=-0.15mm]
    \node [on chain=1,yshift=-2cm,xshift=-.5cm] {$\sturmian{\alpha} = $};
    \node [on chain=1,tape,minimum width=3cm] (imagestart) {};
    \node [on chain=1,tape,minimum width=1cm] (imagezero) {$0^{k} 1$};
    \node [on chain=1,tape] (imageone) {$0$};
    \node [on chain=1,tape,draw=none] {$\ldots$};
\end{scope}

\path[dashed] (start.south west) edge (imagestart.north west)
		  (topzero.south west) edge (imagezero.north west)
		  (topone.south west) edge (imageone.north west)
		  (topone.south east) edge (imageone.north east);
	
\draw [thick, decoration={brace,raise=1mm}, decorate]
      (start.north west) -- node[above,yshift=1mm] {$\sturmian{\beta}[0..m]$} (topone.north east);
		  
\draw [thick, decoration={brace,mirror,raise=1mm}, decorate] 
			(imagestart.south west) -- node[below,yshift=-1mm] {$\sturmian{\alpha}[0..n]$} (imageone.south east);
\end{tikzpicture}
\caption{Simple diagram for $\sturmian{\beta}[m] = 1$}
\label{fig:case1}
\end{minipage}
\end{figure}

\begin{description}
\item[Case] $\sturmian{\beta}[m] = 0$: \\
Clearly $\sturmian{\alpha}[0..n]$ ends with $0^{k} 1 = \morphism{k}(0)$ since $\sturmian{\beta}[m] = 0$. This gives us Figure~\ref{fig:case0}. 
\begin{enumerate}[(i)]
\item Let $u$ be a repetition word in $\sturmian{\beta}$ at position $m$. If $\sturmian{\beta}[0..m-1]$ is a suffix of $u$ then certainly $\sturmian{\alpha}[0..n-1] = \morphism{k}(\sturmian{\beta}[0..m-1])$ is a suffix of $\morphism{k}(u)$.
\begin{figure}[!h]
\centering 
\begin{tikzpicture}[scale=0.5]
\edef\sizetape{6.5mm}
\tikzstyle{tape}=[draw,minimum size=\sizetape]

\begin{scope}[start chain=1 going right, node distance=-0.15mm]
	  \node [on chain=1] {$\sturmian{\beta} = $} ;
    \node [on chain=1,tape,minimum width=1.5cm] (start) {};
    \node [on chain=1,tape] (first0) {$0$};
    \node [on chain=1,tape,minimum width=1.5cm] {$u'$};
    \node [on chain=1,tape] (second0) {$0$};
    \node [on chain=1,tape,minimum width=1.5cm] {$u'$};
    \node [on chain=1,tape] (question) {$?$};
    \node [on chain=1,tape,draw=none] {$\ldots$};
\end{scope}

\begin{scope}[start chain=1 going right, node distance=-0.15mm]
    \node [on chain=1,yshift=-2cm,xshift=-1cm] {$\sturmian{\alpha} = $};
    \node [on chain=1,tape,minimum width=2cm] (imagestart) {};
    \node [on chain=1,tape] (imagefirstL0) {$0^{k}$};
    \node [on chain=1,tape] (imagefirstR0) {$1$};
    \node [on chain=1,tape,minimum width=2cm] {$v'$};
    \node [on chain=1,tape] (imagesecondL0) {$0^{k}$};
    \node [on chain=1,tape] (imagesecondR0) {$1$};
    \node [on chain=1,tape,minimum width=2cm] {$v'$};
    \node [on chain=1,tape] (imagequestion) {$0^{k}$};
    \node [on chain=1,tape,draw=none] {$\ldots$};
\end{scope}

\path[dashed] (start.south west) edge (imagestart.north west)
		  (first0.south west) edge (imagefirstL0.north west)
		  (first0.south east) edge (imagefirstR0.north east)
		  (second0.south west) edge (imagesecondL0.north west)
		  (second0.south east) edge (imagesecondR0.north east)
		  (question.south west) edge (imagequestion.north west);

\draw [thick, decoration={ brace, mirror, raise=1mm }, decorate] 
			(imagefirstR0.south west) -- node[below,yshift=-1mm] {$v$} (imagesecondR0.south west);		
\draw [thick, decoration={ brace, mirror, raise=1mm }, decorate] 
			(imagesecondR0.south west) -- node[below,yshift=-1mm] {$v$} (imagequestion.south east);
\draw [thick, decoration={ brace, raise=1mm }, decorate] 
			(first0.north west) -- node[above,yshift=1mm] {$u$} (second0.north west);
\draw [thick, decoration={ brace, raise=1mm }, decorate] 
			(second0.north west) -- node[above,yshift=1mm] {$u$} (question.north west);
\end{tikzpicture}
\caption{Detailed diagram for $\sturmian{\beta}[m] = 0$}
\label{fig:detailed0}
\end{figure}

Suppose that $u$ is a suffix of $\sturmian{\beta}[0..m-1]$. Since $\sturmian{\beta}[m] = 0$ we know $u$ begins with $0$ and write $u = 0u'$. Since $u'$ is a prefix of $\sturmian{\beta}[m+1..\infty]$, we see that $v' := \morphism{k}(u')$ is a prefix of $\sturmian{\alpha}[n+1..\infty]$. The prefix $u'$ in $\sturmian{\beta}[m+1..\infty]$ is followed by $00$, $01$ or $10$. Since $\morphism{k}(00)$, $\morphism{k}(01)$ and $\morphism{k}(10)$ all start with at least $k$ zeros, we deduce that $v'$ (as it occurs at the beginning of $\sturmian{\alpha}[n+1..\infty]$) is followed by $k$ zeros. Thus, $v := 1 v' 0^{k}$ is a prefix of $\sturmian{\alpha}[n..\infty]$. From the other occurrence of $u$ (as a suffix of $\sturmian{\beta}[0..m-1]$) we deduce that $1 v' 0^{k}$ is also a suffix of $\sturmian{\alpha}[0..n-1]$. We conclude that $v$ is a repetition word in $\sturmian{\alpha}$ at position, and note that $v = 1 v' 0^{k}$ is a conjugate of $0^{k} 1 v' = \morphism{k}(0u') = \morphism{k}(u)$, as required. 

\item Let $v$ be a repetition word in $\sturmian{\alpha}$ at position $n$. The $1$ at position $n$ is preceded by $k$ zeros. Hence, $\sturmian{\alpha}[0..n-1]$ ends in $0^{k}$, so $v$ ends in $0^{k}$. Clearly $v$ begins with $1$, let $v'$ be such that $v = 1 v' 0^{k}$. We do not know whether the trailing $0^{k}$ is the beginning of $\morphism{k}(0)$ or $\morphism{k}(10)$, but in either case $v'$ is $\morphism{k}(u')$ for $u'$ a factor of $\sturmian{\beta}$.

If $\sturmian{\alpha}[0..n-1]$ is a proper suffix of $v$ then $\sturmian{\alpha}[0..n-k-1]$ is a suffix of $v'$. Then $\sturmian{\beta}[0..m-1]$ is a suffix of $u'$, and hence $u := 0u'$ is a repetition word in $\sturmian{\beta}$ at position $m$ such that $v$ is a conjugate of $\morphism{k}(u)$. 

Otherwise, $v$ is a suffix of $\sturmian{\alpha}[0..n-1]$. The trailing $0^{k}$ in this occurrence of $v$ is in the image of $\sturmian{\beta}[m] = 0$. The remaining $1 v'$ must be  preceded by $0^{k}$, and then $0^{k} 1 v'$ is the image of $0u'$, which occurs as a suffix of $\sturmian{\beta}[0..m-1]$. Now we have the situation in Figure~\ref{fig:detailed0}. It follows that $u := 0u'$ is a repetition word, and $v = 1 v' 0^{k}$ is a conjugate of $\morphism{k}(u) = 0^{k} 1 v'$.
\end{enumerate}
\item[Case] $\sturmian{\beta}[m] = 1$: \\
The characteristic Sturmian words we consider start with $0$, so $m \neq 0$. Since $\sturmian{\beta}$ does not contain the factor $11$, we know $\sturmian{\beta}[m-1] = 0$. Therefore $\sturmian{\alpha}[0..n]$ ends in $\morphism{k}(01) = 0^{k} 1 0$, as shown in Figure~\ref{fig:case1}. 

\begin{figure}[!h]
\centering 
\begin{tikzpicture}[scale=0.5]
\edef\sizetape{6.5mm}
\tikzstyle{tape}=[draw,minimum size=\sizetape]

\begin{scope}[start chain=1 going right, node distance=-0.15mm]
	  \node [on chain=1] {$\sturmian{\beta} = $} ;
    \node [on chain=1,tape,minimum width=1.5cm] (start) {};
    \node [on chain=1,tape] (first1) {$1$};
    \node [on chain=1,tape,minimum width=1.5cm] {$u'$};
    \node [on chain=1,tape] (first0) {$0$};
    \node [on chain=1,tape] (second1) {$1$};
    \node [on chain=1,tape,minimum width=1.5cm] {$u'$};
    \node [on chain=1,tape] (second0) {$0$};
    \node [on chain=1,tape,draw=none] {$\ldots$};
\end{scope}

\begin{scope}[start chain=1 going right, node distance=-0.15mm]
    \node [on chain=1,yshift=-2cm,xshift=-1cm] {$\sturmian{\alpha} = $};
    \node [on chain=1,tape,minimum width=2cm] (imagestart) {};
    \node [on chain=1,tape] (imagefirst1) {$0$};
    \node [on chain=1,tape,minimum width=2cm] {$v'$};
		\node [on chain=1,tape,minimum width=1cm] (imagefirst0) {$0^{k}1$};
    \node [on chain=1,tape] (imagesecond1) {$0$};
    \node [on chain=1,tape,minimum width=2cm] {$v'$};
    \node [on chain=1,tape,minimum width=1cm] (imagesecond0) {$0^{k}1$};
    \node [on chain=1,tape,draw=none] {$\ldots$};
\end{scope}

\path[dashed] (start.south west) edge (imagestart.north west)
			(first1.south west) edge (imagefirst1.north west)
			(first1.south east) edge (imagefirst1.north east)
		  (first0.south west) edge (imagefirst0.north west)
		  (second1.south west) edge (imagesecond1.north west)
		  (second1.south east) edge (imagesecond1.north east)
		  (second0.south west) edge (imagesecond0.north west)
		  (second0.south east) edge (imagesecond0.north east);

\draw [thick, decoration={ brace, mirror, raise=1mm }, decorate] 
			(imagefirst1.south west) -- node[below,yshift=-1mm] {$v$} (imagefirst0.south east);
\draw [thick, decoration={ brace, mirror, raise=1mm }, decorate] 
		  (imagesecond1.south west) -- node[below,yshift=-1mm] {$v$} (imagesecond0.south east);
\draw [thick, decoration={ brace, raise=1mm }, decorate]
		  (first1.north west) -- node[above,yshift=1mm] {$u$} (first0.north east);
\draw [thick, decoration={ brace, raise=1mm }, decorate]
		  (second1.north west) -- node[above,yshift=1mm] {$u$} (second0.north east);
\end{tikzpicture}
\caption{Detailed diagram for $\sturmian{\beta}[m] = 1$}
\label{fig:detailed1}
\end{figure}

\begin{enumerate}[(i)]
\item Suppose $u$ is a repetition word in $\sturmian{\beta}$ at position $m$, and let $v := \morphism{k}(u)$. We know that $\morphism{k}(\sturmian{\beta}[0..m-1]) = \sturmian{\alpha}[0..n-1]$ and $\morphism{k}(\sturmian{\beta}[m..\infty]) = \sturmian{\alpha}[n..\infty]$. Thus, 
\begin{itemize}
\item $v$ is a prefix of $\sturmian{\alpha}[n..\infty]$ if $u$ is a prefix of $\sturmian{\beta}[m..\infty]$
\item $v$ is a suffix of $\sturmian{\alpha}[0..n-1]$ if $u$ is a suffix of $\sturmian{\beta}[0..m-1]$
\item $\sturmian{\alpha}[0..n-1]$ is a suffix of $v$ if $\sturmian{\beta}[0..m-1]$ is a suffix of $u$.
\end{itemize} 
It follows that $v$ is a repetition word in $\sturmian{\alpha}$ at position $n$. 
\item Suppose $v$ is a repetition word in $\sturmian{\alpha}$ at position $n$. We know $v$ starts with $0$ since $\sturmian{\alpha}[n] = 0$, and $v$ ends with $1$ since $\sturmian{\alpha}[n-1] = 1$, therefore $v = 0 v' 0^{k} 1$ for some $v'$. Then $v' = \morphism{k}(u')$ for some $u'$, and we define $u := 1 u' 0$ so that
$$
\morphism{k}(u) = \morphism{k}(1 u' 0) = 0 v' 0^{k} 1 = v.
$$
It is also clear that 
\begin{itemize}
\item $u$ is a prefix of $\sturmian{\beta}[m..\infty]$
\item $u$ is a suffix of $\sturmian{\beta}[0..m-1]$ if $v$ is a suffix of $\sturmian{\alpha}[0..n-1]$
\item $\sturmian{\beta}[0..m-1]$ is a suffix of $u$ if $\sturmian{\alpha}[0..n-1]$ is a suffix of $v$,
\end{itemize} 
so we conclude that $u$ is a repetition word in $\sturmian{\beta}$ at position $m$. 
\end{enumerate}
\end{description} 
\end{proof}

\begin{theorem}
\label{theorem:main}
Let $\alpha$ be a directive sequence and let $\beta := \alpha[1..\infty]$. Let $n \geq 0$ be a nonnegative integer. Let $t$ be the number of trailing zeros in $\ostrowski{\alpha}{n+1}$. Then $r_{\alpha}(n)$ is a conjugate of $X_{t}$, except when all of the following conditions are met:
\begin{itemize}
\item The last nonzero digit in $\ostrowski{\alpha}{n+1}$ is $1$.
\item $\ostrowski{\alpha}{n+1}$ contains at least two nonzero digits.
\item The last two nonzero digits of $\ostrowski{\alpha}{n+1}$ are separated by an even number of zeros.
\end{itemize}
When $\ostrowski{\alpha}{n+1}$ meets these conditions, then $r_{\alpha}(n)$ is a conjugate of $X_{t+1}$ .
\end{theorem}
\begin{proof}
Let $d_{k} \cdots d_{0} = \ostrowski{\alpha}{n+1}$ be the Ostrowski representation of $n+1$. Let $t$ be the number of trailing zeros in $\ostrowski{\alpha}{n+1}$. We use induction on $t$ to prove that $r_{\alpha}(n)$ is a conjugate of $X_{t}$, or under the conditions described above, a conjugate of $X_{t+1}$.

\begin{description}
\item[Base case $t = 0$:]
Since $n+1 > 0$, we have $d_0 > 0$. By Theorem~\ref{thm:decomp}, we have 
$$\sturmian{\alpha}[0..n] = X_{k}^{d_{k}} \cdots X_{0}^{d_0}.$$
If $d_{0} \geq 2$ then we are done since $\sturmian{\alpha}[0..n]$ ends in $00$. Hence $\sturmian{\alpha}[n-1] = \sturmian{\alpha}[n] = 0$ and $r_{\alpha}(n) = 0 = X_0$ is the shortest repetition word at position $n$. Let us assume without loss of generality that $d_{0} = 1$. 

According to the induction hypothesis, the second last nonzero digit in $\ostrowski{\alpha}{n+1}$ becomes relevant when the last nonzero digit is 1. If $d_0$ is the only nonzero digit, then $n = 0$ and $r_{\alpha}(0)$ is clearly $\sturmian{\alpha}[0] = 0$. Otherwise, pick $\ell > 0$ minimal such that $d_{\ell} \neq 0$. That is, let $d_{\ell}$ be the second last nonzero digit. Note that by Theorem~\ref{thm:decomp}, the word $\sturmian{\alpha}[0..n-1]$ ends in $X_{\ell}$.

If $\ell$ is even then $X_{\ell}$ ends in $0$ (by a simple induction), so $\sturmian{\alpha}[n-1] = 0$ and it follows that $r_{\alpha}(n) = 0$. When $\ell$ is odd, the word $X_{\ell}$ ends in $X_1$ and $X_1$ ends in $1$. It follows that
$$\sturmian{\alpha}[0..n-1] = \morphism{\alpha[0]}(\sturmian{\beta}[0..m-1])$$
for some $m \geq 0$. We claim that $\sturmian{\beta}[m] = 0$, since otherwise
$$\morphism{\alpha[0]}(\sturmian{\beta}[0..m]) = \sturmian{\alpha}[0..n]$$
so Lemma~\ref{lemma:image} states that $\ostrowski{\alpha}{n+1}$ ends in $0$, contradicting $d_0 = 1$. Then $\sturmian{\alpha}[n..\infty]$ begins with $\morphism{\alpha[0]}(\sturmian{\alpha}[m]) = X_1$, so $r_{\alpha}(n) = X_1$. 

\item[Inductive step $t > 0$:]
We note that removing (or adding) trailing zeros from $\ostrowski{\alpha}{n+1}$ does not change whether it satisfies all three conditions in the theorem. We will assume that $\ostrowski{\alpha}{n+1}$ does not meet the conditions, since the proof is nearly identical if it does meet the conditions. 

Let $\{ X_i \}_{i=0}^{\infty}$ and $\{ Y_i \}_{i=0}^{\infty}$ be standard sequences corresponding to the directive sequences $\alpha$ and $\beta$ respectively. Lemma~\ref{lemma:image} states that $\sturmian{\alpha}[0..n] = \morphism{\alpha[0]}(\sturmian{\beta}[0..m])$ where $m \geq 0$ is such that $$\ostrowski{\beta}{m+1} = d_{k} \cdots d_{1}.$$
Note that $d_{k} \cdots d_{1}$ has $t-1$ trailing zeros, so $r_{\beta}(m)$ is a conjugate of $Y_{t-1}$ by induction. By Lemma~\ref{lemma:recursivecase}, $r_{\alpha}(n)$ is a conjugate of $\morphism{\alpha[0]}(Y_{t-1}) = X_{t}$, completing the proof. 
\end{description}
\end{proof}

Let us continue our example with a directive sequence $\alpha$ starting with $1,3,2,2$. Recall that
\begin{align*}
\sturmian{\alpha} &= 01010 \, 10010 \, 10100 \, 10101 \, 01001 \, 01010 \, 01010 \, 1010 \cdots
\end{align*}
Consider the shortest repetition words at positions 23 through 26. These positions happen to give illustrative examples of the theorem. 
\begin{align*}
r_{\alpha}(23) &= 0                & X_0 &= 0                & \ostrowski{\alpha}{24} = 1101 \\
r_{\alpha}(24) &= 1010100          & X_2 &= 0101010          & \ostrowski{\alpha}{25} = 1110 \\
r_{\alpha}(25) &= 01               & X_1 &= 01               & \ostrowski{\alpha}{26} = 1111 \\
r_{\alpha}(26) &= 10               & X_1 &= 01               & \ostrowski{\alpha}{27} = 1120
\end{align*}
When $n = 23$, there are no trailing zeros in $\ostrowski{\alpha}{24} = 1101$ and we have an odd number of zeros between the last two nonzero digits. Hence, $r_{\alpha}(23)$ is a conjugate of $X_0 = 0$. Compare this to $n = 25$, where $\ostrowski{\alpha}{26} = 1111$ also has no trailing zeros, but the last two ones are adjacent, so $r_{\alpha}(25)$ is a conjugate of $X_1$. We are in a similar situation for $n = 24$, but with an trailing zero so $r_{\alpha}(24)$ is a conjugate of $X_2$. Finally, consider $n = 26$ where the last two nonzero digits are adjacent and we have a trailing zero, like $n = 24$, but the last nonzero digit is not a one. It follows that $r_{\alpha}(26)$ is a conjugate of $X_1$. Although $r_{\alpha}(25)$ and $r_{\alpha}(26)$ are both conjugates of $X_1$, they are not the same. 

\section{Open Problems and Further Work}

It would be interesting to generalize the result to two-sided Sturmian words, with an appropriate definition for local period in two-sided words. We might define a repetition word in $w \in {}^\omega \Sigma^\omega$ at position $n$ as a word that is simultaneously a prefix of $w[n..\infty]$ and a suffix of $w[-\infty..n-1]$. Note that if we extend a characteristic Sturmian word $\sturmian{\alpha}$ to a two-sided word $w$, the local period at position $n$ in $w$ may not be the same as the local period at position $n$ in $\sturmian{\alpha}$. 

Our main result is about the local period and the shortest repetition word, but Lemma~\ref{lemma:recursivecase} applies to all repetition words at a specific position. Is it possible to extend our result to all repetition words, not just the shortest repetition word? Patterns in the lengths of repetition words for the Fibonacci word suggest that it is possible, but we do not have a specific conjecture.


\begin{thebibliography}{9}
  
\bibitem{shallit}
  J. Shallit
  (Personal communication).
  
\bibitem{shallitallouche}
  J.-P. Allouche, J. Shallit, {\it Automatic Sequences}, Cambridge Univ. Press (2003).
  
\bibitem{restivomignosi}
	F. Mignosi, A. Restivo,
	{\it Characteristic Sturmian words are extremal for the Critical Factorization Theorem},
	Theoret. Comput. Sci. {\bf 454}
	(2012), 199--205.
	
\bibitem{covenhedlund}
	E. Coven, G. Hedlund, 
	{\it Sequences with minimal block growth},
	Math. Systems Theory 7 (1973) 138--153.

\bibitem{berstel}	  
	J. Berstel and P. S\'e\'ebold,
	Sturmian words,
	in M. Lothaire, ed.,
	{\it Algebraic Combinatorics on Words},
	Encyc. of Math. and its Appl., Vol.\ 90,
	Cambridge Univ. Press (2002) 45--110.
	  
\end{thebibliography}
\end{document}